\documentclass[pdflatex,sn-mathphys-num]{sn-jnl}

\usepackage{float}
\usepackage{graphicx}%
\usepackage{multirow}%
\usepackage{amsmath,amssymb,amsfonts}%
\usepackage{amsthm}%
\usepackage{mathrsfs}%
\usepackage[title]{appendix}%
\usepackage{xcolor}%
\usepackage{textcomp}%
\usepackage{manyfoot}%
\usepackage{booktabs}%
\usepackage[linesnumbered,ruled]{algorithm2e}
\usepackage{algpseudocode}
\usepackage{listings}%

\newtheorem{lem}{Lemma}
\theoremstyle{thmstyleone}%
\newtheorem{theorem}{Theorem}
\theoremstyle{thmstyletwo}%

\theoremstyle{thmstylethree}%

\begin{document}

\title[An average case efficient algorithm for solving two-variable linear Diophantine equations]{An average case efficient algorithm for solving two-variable linear Diophantine equations}


\author*[1]{\fnm{Mayank} \sur{Deora}}\email{mbdeora@gmail.com}

\author[2]{\fnm{Pinakpani} \sur{Pal}}\email{pinak@isical.ac.in}

\affil*[1]{\orgdiv{Electronics and Communication Sciences Unit}, \orgname{Indian Statistical Institute}, \orgaddress{ \city{Kolkata}, \country{India}}}

\affil[2]{\orgdiv{Electronics and Communication Sciences Unit}, \orgname{Indian Statistical Institute}, \orgaddress{ \city{Kolkata},  \country{India}}}


\abstract{Solving two-variable linear Diophantine equations has applications in many cryptographic protocols such as RSA and Elliptic curve cryptography. The Extended Euclid’s algorithm is a well known algorithm to solve these equations. We revisit two algorithms to solve two-variable linear Diophantine equations. We write the iterative version of one of the revisited algorithms. For another, we do a fine-grained analysis of the number of recursive calls and arrive at a periodic function that represents the number of recursive calls. We find the period and use it to derive multiple upper bounds on the average number of recursive calls. The upper bounds conclude that the average number of recursive calls in the analyzed algorithm is a constant term improvement over the Extended Euclid's algorithm. We propose an iterative variant of the algorithm. We implement this algorithm and find that the average number of iterations by our algorithm is less than that of two existing algorithms. We compare the number of iterations in this algorithm and in Extended Euclid's algorithm. In the comparison, we find that for 100\% of the solvable instances of inputs, the iterations are less than that in Extended Euclid's algorithm.}


%

\maketitle
\section{Introduction}
 Diophantine equations are equations that have only integer solutions. Solving a two-variable linear Diophantine equation $ax+by=c$, requires determining integers $x$ and $y$ that satisfy this equation. In public key cryptography, especially in RSA and Elliptic curve cryptography, we need to repetitively compute the modulo multiplicative inverse \cite{stallings2006cryptography}. Finding modulo multiplicative inverse $a^{-1} \bmod b$ is essentially equivalent to solving the Diophantine equation $ax-by=1$ and finding $x \bmod b$.  \par In \cite{deora2023algorithm}, a recursive algorithm has been proposed to solve two-variable linear Diophantine equations referred to as the DEA-R algorithm in this paper. Euclidean algorithm and binary gcd algorithms are used for computing the greatest common divisor $(gcd)$ of two integers, which can be extended to solve two-variable linear Diophantine equations.  
 In \cite{sumarti2023method} and \cite{kumar2022alternative}, as a case study, these kinds of Diophantine equations are solved by reducing the equations to an instance of an integer linear programming (ILP) problem. ILP problem is an NP-complete problems. Since the two-variable linear Diophantine equations can be solved by Extended Euclid's algorithm \cite{cormen2009introduction} in polynomial time, so  reducing them to an NP-complete problem is not required. Aardal et. al. \cite{aardal2000solving} and Ramachandran \cite{ramachandran2006use} proposed algorithms for a system of linear Diophantine equations when there are upper bounds and lower bounds on variables. They also deal with an NP-Complete problem. 
 
 \par There are algorithms for solving a system of linear Diophantine equations in \cite{chou1982algorithms},  \cite{lazebnik1996systems}, \cite{bradley1971algorithms}, and \cite{esmaeili2001class}. Lazebnik \cite{lazebnik1996systems} provides a survey on a theorem for the solvability of a system of Diophantine equations, along with the construction of the solution. It does not specifically provide an algorithm. Bradley \cite{bradley1971algorithms} used Hermite normal form and Smith normal form of a matrix to develop an algorithm for solving a system of linear Diophantine equations. In \cite{chou1982algorithms}, there is an algorithm for computing solutions of a system, that needs solution of two-variable linear Diophantine equations in a subroutine.  For solving a system, Esmaeili et. al.  \cite{esmaeili2001class} use a method generalized to solve linear algebraic systems. We can reduce a two-variable linear Diophantine equation to an instance of a problem that is solvable by the algorithms for a system of these equations, but this reduction makes the problem more complex. 

 \par Lankford \cite{lankford1989non} and  Ajili et. al. \cite{ajili1995complete} proposed algorithms to solve for non-negative integer basis of the linear Diophantine equations.  Lankford \cite{lankford1989non} mentions that finding the time complexity of their proposed algorithm is a hard open problem. There is an algorithm in \cite{clausen1989efficient} that finds complete information of the set of all non-negative solutions of linear homogeneous or non-homogeneous Diophantine equations. In \cite{clausen1989efficient}, although there is no discussion on  the time complexity of the proposed algorithm, one can infer from its implementation that the time complexity is exponential in the size of the inputs. The algorithms proposed in \cite{clausen1989efficient}  and \cite{lankford1989non}  are not efficient for two-variable linear Diophantine equations in general, but they are applicable in certain other problems such as associative commutative unification.

 \par The number of recursions or steps in the average case in the Euclidean and binary gcd algorithm is the same ($O(\log b)$) \cite{knuth2014art}, where $b=min(a,b)$. The average number of recursions in Extended Euclid's algorithm is $\Theta(\log b)$ \cite{cormen2009introduction}, \cite{hensley1994number}, \cite{heilbronn1969average}, \cite{dixon1970number}.   In \cite{deora2023algorithm}, it was proved that the DEA-R algorithm will always incur fewer than or an equal number of recursive calls compared to Extended Euclid's algorithm.  The average number of recursive calls was analyzed to determine whether there is any asymptotic or constant-time improvement in the average number of recursive calls over Extended Euclid's algorithm \cite{cormen2009introduction}. 
 For finding the average number of recursive calls, there was an assumption that if we want to solve $ax+by=c$, then $c$ can take values between 1 to a natural number $N$. This approach is constrained by the fact that there is no limit on the value of $N$, which essentially means that $N$ can be $\infty$. As a result, no  improvement  was found in the average number of recursive calls over Extended Euclid's algorithm. For the DEA-R algorithm, there is no closed-form expression for the average number of recursive calls or a bound that is a function of the inputs to the algorithm. To find the closed-form expression, we need to compute the cardinality of some special integer sets, for which \cite{deora2023algorithm} provides no discussion. Till now, the DEA-R algorithm has not been compared with the Extended Euclid's algorithm on a large set of large-sized inputs.

 \par Our contributions in this work are as follows:
 \begin{enumerate}
     \item Implementing a recursive function has an inherent run-time overhead in terms of space while being implemented on a computer. For the implementation of the DEA-R algorithm, we develop an iterative version (DEA-I algorithm) of the recursive algorithm (DEA-R algorithm) in Section \ref{secDEAI}. We implement it on a large number of inputs. Each input is 4096 bits in size. 
     \item In Section \ref{secPNRC} we define a function $R_{a,b}:c \to \mathbb{N}$, which maps $c$ to the number of recursive function calls to solve the Diophantine equation ($ax+by=c$) using the DEA-R algorithm. We propose a theorem (theorem \ref{theo10}) to prove that the function $R_{a,b}$ is periodic. In theorem \ref{theo10} we prove that the fundamental period of $R_{a,b}$ is the lowest common multiple of all the intermediate integers (except the first integer) that are passed in as arguments to subsequent recursive calls, while running Euclid's gcd algorithm. 
     \item In Section \ref{secACAD}, we derive various upper bounds on the average number of recursive calls by the DEA-R algorithm.  We find that DEA-R algorithm does a constant improvement in number of recursive calls over Extended Euclid's algorithm in the average case.
    
      \item In Section \ref{secIR}, there is an implementation and comparison of the DEA-I algorithm with two iterative versions of the Extended Euclid's algorithm on a large set of large-sized random inputs. We find that the DEA-I algorithm outperforms in terms of the average number of iterations. We also find that for 100\% of the solvable input instances ($a,b,c$), DEA-I algorithm is better than Extended Euclid's algorithm. In another implementation, we verify that there is constant-term improvement in the average number of iterations (recursions) in DEA-I (DEA-R) algorithm.
 
 \end{enumerate}

  \par The subsequent sections of this paper are organized as follows. In Section \ref{secNot}, we introduce mathematical notations used throughout the paper.  Section \ref{secRAlgo} provides an overview of the DEA-R and EEA-R algorithms and the improvement of the DEA-R over the EEA-R algorithm. EEA-R algorithm is the recursive Extended Euclid's algorithm, which is given in \cite{cormen2009introduction}. In Section 3, we reduce a recursive version of the EEA-R algorithm to an iterative version. We refer to this algorithm as EEA-I, which is used in comparison in Section \ref{secIR}. Section \ref{secDEAI} presents iterative version (DEA-I algorithm) of the recursive algorithm (DEA-R algorithm). The proof of periodicity of the number of recursive function calls in DEA-R is presented in Section \ref{secPNRC}.
  Section \ref{secACAD} presents the average case analysis of DEA-R, along with a theorem that is applicable in the analysis of the algorithm. Section \ref{secIR} reports the implementation and results of comparing the DEA-I with the extended Euclid's algorithm. In Section \ref{secCon}, we conclude the paper and discuss possible future works.

\section{Notations}\label{secNot}

 Throughout this paper, we write recursive function call or recursive call to refer to a call to a recursive function. Both the terms represent the same kind of function here, because the function (in DEA-R or in EEA-R) that we refer to is recursive. For the convenience of the readers, in table \ref{tab:my_label} we list the various notations used.
\begin{table}[h]
\caption{Notations for variables and abbreviations in this paper}
    \centering

    \begin{tabular}{p{4cm}|p{7.6cm}}
         \textbf{Notation} &  \textbf{Description} \\\hline
         $a_1,a_2$ & First and second arguments respectively in the input ($a_1,a_2,c$) given to the DEA-R algorithm to solve $a_1x+a_2y=c$. It also refers to the first and second arguments, respectively, in the input ($a_1,a_2$) given to Extended Euclid's algorithm. \\\\
         $a_i\,  (1\le i \le k+1)$ & Extension of previous notations. With the assumption that Extended Euclid's algorithm incurs $k+1$ recursive calls on an input $(a_1,a_2)$, the input to the next recursive call is $(a_2,a_3)$. Similarly, the inputs to its subsequent recursive calls are $(a_3,a_4),(a_4,a_5),......,(a_k,a_{k+1}),(a_{k+1},0)$ respectively.
    \\\\ $q_i\,  (1\le i \le k)$ & Quotient of the division $\frac{a_i}{a_{i+1}}$, where $a_i$ and $a_{i+1}$ are defined in the previous notation.
    \\\\$a_i=q_ia_{i+1}+a_{i+2}\,$ $ (1\le i \le k)$ & Derived from the previous notations. $a_{i+2}$ is the remainder of the division $\frac{a_i}{a_{i+1}}$. By the design of Extended Euclid's algorithm $a_{k+1}$ is the $gcd$ of $a_1,a_2$ and $a_{k+2}=0$. \\\\
    $\mathbf{c_i}\,  (1\le i \le k)$ & A set of integers $c$, such that $c=a_i+Q_ia_{i+1}$, where $Q_i \in \mathbb{Z}$ i.e. all members $c$ of $c_i$ satisfy $c=a_i+Q_ia_{i+1}$. 
    \\\\
    $gcd(a,b)$ & Greatest common devisor ($gcd$) of $a$ and $b$ 
    \\\\ 
    $lcm(a_2,a_3,.....,a_{k+1})$ & lowest common multiple ($lcm$) of the integers $a_2,a_3,.....,a_{k+1}$\\\\
    
DEA-R & Diophantine Equations Algorithm - Recursive, proposed in \cite{deora2023algorithm}\\\\
EEA-R & Extended Euclid's algorithm (Recursive) given in \cite{cormen2009introduction}
\\\hline
 \end{tabular}
    
    \label{tab:my_label}
\end{table}
 
\par We denote the solution of Diophantine equation $ax+by=c$ by $(x,y)$ and general solution by $(x+m\frac{b}{gcd(a,b)}, y-m\frac{a}{gcd(a,b)})$ for any integer variable $m$. It is well known that, if a linear Diophantine equation is solvable, it has infinite solutions. In the general solution, $m$ varies and thus the general solution represents all the solutions of the Diophantine equation.

\section{Revisited algorithms}\label{secRAlgo}
Here, we give an overview of the  DEA-R algorithm and the EEA-R algorithm. We also revisit the theoretical comparison of DEA-R  with Extended Euclid's algorithm (EEA-R), which was given in \cite{deora2023algorithm}. At the end of this section, there is EEA-I algorithm which is an iterative variant of EEA-R algorithm.  
\subsection{DEA-R algorithm} The Algorithm 1 describes the DEA-R algorithm with a minor change in the print statement at line 3. 

\begin{algorithm}
\caption{\textbf{DEA-R} Algorithm to solve 2-variable Linear Diophantine equation \cite{deora2023algorithm}}
 \KwIn{ Integers $a,b,c$ ($a>b ,a\ne 0,b\ne 0,c\ne 0$)}  

\KwOut{Integer $y$}

\SetKwFunction{FMain}{f}
\SetKwProg{Fn}{Function}{:}{}

\Fn{\FMain{$a$, $b$,$c$}}{
   
\If{(b=0)} 
 {PRINT ("$a$ is the gcd of original inputs. $c$ is not a multiple of $g$")\;
 \textbf{exit}\;
}
\ElseIf {$((c-a) \bmod b=0)$} 
{
$y\gets (c-a)/ b$\;
\Return $y$\;
}
\Else 
{
$y\gets\left(c-f(b,a \bmod b,c)\times a)\right/b$\;
\Return $y$\;
 }
}\end{algorithm}
In DEA-R, the function $f$ takes $a,b,c$ as inputs, where $a>b,\,\, a\ne0, b\ne 0$ and it outputs $y$. After this  we can compute $x$, such that $ax+by=c$ as follows:
$$x=\frac{c-by}{a}$$

When $c=0$, the equation has trivial solutions. When $c\ne 0$, we summarize the DEA-R algorithm as follows:

\begin{enumerate}
    \item If $b=0$, no solution is possible. $b=0$ implies that $c$ is not a multiple of $g$ ($g$ is gcd of $a$ and $b$). So, $ax+by=c$ does not satisfy Bazout's lemma \cite{seroul2012programming}. Also, $a$ will be the gcd of inputs $(a,b)$ of the algorithm (input to the initial recursive function call). This output is similar to the output in Extended Euclid's algorithm \cite{cormen2009introduction}, where $b=0$ is the only base condition and at this step, gcd is produced in output. This output is different than the output given in the previous version of this algorithm.
    \item If $(c-a) \bmod b=0$, then output $(c-a)/b$ else go to step 3.
    \item Call function $f(b,a \bmod b,c)$ and output \newline $\frac{c-f(b,a \bmod b,c)\times a}{b}$
\end{enumerate}
    We give a brief note on the reason for the base condition $b=0$, implying no solution given in \cite{deora2023algorithm} also. Let $c$ be a multiple of $g$, and the arguments in the second-to-last recursive call be $a',b',c$. Then $(c-a')$ must be a multiple of $b'$. because $b'= a$ is same as $g$ and $a'$ and $c$ both are multiple of $g$. Thus, it is a contradiction because the DEA-R algorithm makes the last recursive call with arguments ($g,0,c$).
\subsection{EEA-R algorithm}

EEA-R algorithm \cite{cormen2009introduction} is a recursive variant of Extended Euclid's algorithm. The EEA-R algorithm solves for an integer solution of the following Diophantine equation.
$$ax+by=gcd(a,b)$$
To solve a Diophantine equation $ax'+by'=c$ for $x'$ and $y'$ with any value of $c$, we check whether $c$ is divisible by $gcd(a,b)$. If $c$ is not divisible then $ax'+by'=c$ is not solvable, otherwise we compute $x$ and $y$ as follows:
$$x'=\frac{xc}{gcd(a,b)}, y'=\frac{yc}{gcd(a,b)}$$

\begin{algorithm}
\caption{EEA-R algorithm to solve 2-variable linear Diophantine equations \cite{cormen2009introduction}}

\KwIn{Integers $a,b$ ($a>b$)}  
\KwOut{Integers $x,y$ such that $ax+by=gcd(a,b)$}

\SetKwFunction{FMain}{f}
\SetKwProg{Fn}{Function}{:}{}
\Fn{\FMain{$a$, $b$}}{
\If{$b=0$}
{\Return ($a,1,0$)\;} 
\Else {$(g',x_{old},y_{old})\gets f(b,a \bmod b)$\;
 $g \gets g'$\;
$x\gets y_{old}$\;
 $y \gets x_{old}-\lfloor {a/b}\rfloor y_{old} $\;
\Return ($g,x,y$)\;
}
    }
\end{algorithm}

\subsection{Comparison of DEA-R with EEA-R} 
In \cite{deora2023algorithm}, there is a comparison of the number of recursive function calls and time complexity of the DEA-R algorithm with that of EEA-R. In the worst case, the number of recursive function calls in both algorithms are equal. In DEA-R, if the size of $c$ in input is the same as the size of $max(a,b)$, $c$ may account equally for the time complexity of the DEA-R algorithm. So, when the size of $c$ in input is equal to the size of $max(a,b)$, in the worst case, asymptotic time complexity of DEA-R algorithm is same as that of Extended Euclid's algorithm. In \cite{deora2023algorithm}, there is no consideration for the size or bit-size of $c$ for comparison. So, the comparative analysis has an implicit assumption that the bit-size of $c$ is less than or equal to the maximum of the bit sizes of $a$ or $b$. For some values of $c$ (third parameter in input of DEA-R), the number of recursive calls is less than that in Extended Euclid's algorithm.

In DEA-R, there are more basic arithmetic operations (division, multiplication, subtraction, and comparison) in one iteration as compared to Extended Euclid's algorithm. Assume that $R_{EE}$ and $R_{A}$ denote the number of recursive function calls of Extended Euclid's algorithm and DEA-R, respectively. It is explained in \cite{deora2023algorithm}, that always $R_{EE} \ge R_A$. It is also given that the criteria for DEA-R being more efficient is that the difference between the number of recursive calls in both the algorithms must be greater than $\frac{1}{3}$ of the number of recursive calls in the Extended Euclid's algorithm. Thus, DEA-R will be more efficient than Extended Euclid's algorithm only when the following relation holds:

$$R_{EE}-R_{A}>\frac{1}{3}R_{EE}$$
or \begin{equation} \label{eq8}
   \frac{2}{3} R_{EE} >  R_A 
\end{equation}
i.e., the number of recursive calls in DEA-R must be less than two-thirds of that in extended Euclid's algorithm. In this paper, we claim that condition (\ref{eq8}) is applicable only when the bit size of $|c|\le max(|a|,|b|)$. By analysis of the different values of $c$ we find that the $c$ values that take the same number of recursive function calls for a fixed value of $a$ and $b$ are expressible as $c=a_i+Q_ia_{i+1}$, where $a_i$'s are explained in notations section. In section 4, we observe the number of recursive function calls as a function of the value of $c$. We see that this function is periodic, and we find the period as well.

\subsection{Iterative version of EEA-R} Extended Euclid's algorithm is a well-known algorithm to solve two-variable linear Diophantine equations. We develop an iterative version of Extended Euclid's algorithm from its recursive version given in the book \cite{cormen2009introduction} (referred to as EEA-R), and we implement it as a C program. We refer to the iterative version as EEA-I. We develop the EEA-I algorithm because recursion has a run-time overhead, and we want to compare the EEA-R and DEA-I algorithms. In section \ref{secIR}, there are results of a comparison between DEA-I and EEA-I. EEA-I has a function $f(a,b)$ which computes integers $x'$ and $y'$ such that $ax'+by'=gcd(a,b)$. If $c$ is divisible by $gcd(a,b)$, we compute $y$ as $y=(\frac{c}{gcd})\times y'$ and $x$ as $x=(\frac{c}{gcd})\times x'$ to solve for integral solutions of $ax+by=c$. To observe an equivalence between recursion in the EEA-R algorithm and iteration in the EEA-I algorithm, we can see that if the first while loop in the EEA-I algorithm incurs $k$ number of iterations, then the number of recursions by the EEA-R algorithm will be $k+1$. This is because line no. 10, 12, and 13 of the EEA-I algorithm are part of the last recursive call of the EEA-R algorithm.

\begin{algorithm}[H]
 \caption{EEA-I}\label{AlgEEA2}

\KwIn{ Integers $a,b,c$ ($a>b ,a\ne 0,b\ne 0,c\ne 0$)}  
\KwOut{Integers $x,y$ such that $ax+by=c$} 


\SetKwFunction{FMain}{f}
\SetKwProg{Fn}{Function}{:}{}
\Fn{\FMain{$a$, $b$, $c$}}{

 $arraysize \gets 0$ \;
\While{$(b \ne 0)$}
{
$floorArray[arraysize]\gets floor(a/b)$\;
 $temp\gets a$\;
 $a\gets b$ \;
 $b \gets temp\bmod b$\;
$arraysize \gets arraysize+1$\;
}

 $gcd \gets a$\;
\If{$c \bmod gcd = 0$}{
$x \gets 1$\;
$y \gets 0$\;
\While{$(arraysize \ge 1)$}
{
$temp \gets x$\;
 $x \gets y$\;

$y \gets temp-floorArray[arraysize-1]*y$ \;
$arraysize \gets arraysize-1$\;
}

$y \gets (c/gcd)\times y$\;
$x \gets (c/gcd)\times x$\;
 \Return $(x,y)$\;
 }
\Else 
{
PRINT ("$c$ is not a multiple of $g$")\;
}
}

\end{algorithm}

In the next section, we propose an iterative version of the DEA-R algorithm (DEA-I algorithm).
\section{DEA-I: an iterative version of DEA-R}\label{secDEAI} Here we summarize the DEA-I algorithm, which is an iterative version of DEA-R algorithm. Note that DEA-I is using the same logic as DEA-R. Here we give a short proof of the correctness of DEA-I.

\begin{algorithm}
\caption{\textbf{DEA-I} Iterative version of algorithm 1 to solve 2-variable Linear Diophantine equation}
\KwIn{Integers $a,b,c$ ($a>b ,a\ne 0,b\ne 0,c\ne 0$) }
\KwOut{Integer $y$}  

\SetKwProg{Fn}{Function}{:}{}
\Fn{\FMain{$a$, $b$, $c$}}{

 $no\_solution\gets 0$\;
 $arraysize\gets 0$ \;
\While{($(c-a) \bmod b \ne 0)$}
{
$coefarray[arraysize]\gets a$\;
$a\gets b$ \;
$b \gets (coefarray[arraysize]) \bmod b$ \;
$arraysize\gets arraysize+1$ \;
\If{($b=0$)}{
 $no\_solution\gets 1$\;
 PRINT ("$a$ is the gcd of original inputs. $c$ is not a multiple of $g$")\;
\textbf{Break}\;
}
}
\If{($no\_solution=0$)}{
 $coefarray[arraysize]\gets a$\;
 $y\gets \frac{c-a}{b}$\;
\While{$(arraysize \ge 1)$}
{
 $y \gets \frac{c-y*coefarray[arraysize-1]}{coefarray[arraysize]}$\;
$arraysize \gets arraysize-1$\;

}
}
\Return $y$\;

}

\end{algorithm}

In DEA-R algorithm, the initial and each subsequent recursive call receives the input tuples $(a_i,a_{i+1},c)$. The notations, $a_i$ and $a_{i+1}$ are described in table \ref{tab:my_label} of notations.
In DEA-I, we have used an array ($coefarray$) to store the intermediate inputs ($a_i$'s). The loop termination condition for the first while loop, at line 4 of DEA-I is equivalent to the termination condition at line 6 of DEA-R. The first while loop of DEA-I has the following loop invariant:
$$
coefarray[i]=a_i
$$
At line 9, in the DEA-I algorithm, if $b=0$ is true,then Diophantine equation has no solution. The $y$ variable is used to store the value output by $f(a,b,c)$.  For a solvable Diophantine equation, if the recursion in the DEA-R algorithm runs for $k$ times, there will be $k$ iterations of the first while loop in the DEA-I algorithm. Assume that in an iteration of the while loop at line 18 of the DEA-I algorithm, the $arraysize$ variable has the value $i$ and the value stored in $y$
represents value returned by function $f(a_i,a_{i+1},c)$ of DEA-R algorithm. Thus $y$ variable at line 17 of the DEA-I algorithm is the same as the value of $f(a_k,a_{k+1},c)$, returned by line 8 in the DEA-R algorithm. Then $y$ variable inside the while loop at line 18 of the DEA-I algorithm satisfies the following loop invariant:
$$
y=\frac{c-f(a_{i+1},a_{i+2},c)\times a_i}{a_{i+1}}
$$
The second loop runs for $k$ times, which implies that the $y$ value returned by $f(a,b,c)$ in the DEA-R algorithm and the value returned by $f(a,b,c)$ in DEA-I are the same.
The DEA-R algorithm is a recursive algorithm. Recursive functions have a runtime overhead when we execute them on a computer in terms of space and time. To remove this runtime overhead, we develop the DEA-I algorithm. In the section on implementations, we evaluate the performance of the DEA-I or DEA-R algorithm and compare them against the widely used Extended Euclid's algorithm.  
\par The periodicity analysis and average case analysis of the number of recursive calls by DEA-R, given in subsequent sections, is also applicable to the number of iterations of the first while loop in DEA-I. We explain a mapping between the number of recursions and the number of iterations as follows. The recursive function in the DEA-R algorithm has the second base condition (at line 6), that is the same as the condition for the termination of the first while loop of DEA-I. The first base condition of the recursive function in the DEA-R algorithm and the condition of the if statement in line 9 of the DEA-I algorithm are the same. Thus, it is easy to observe that for solvable Diophantine equations, the number of recursive calls of the recursive function in DEA-R and the number of iterations of the first while loop in the DEA-I algorithm are equal. For the unsolvable Diophantine equations, the DEA-R algorithm will incur one more recursive call than the number of iterations taken by the first while loop of DEA-I. In the unsolvable case, the if block at line 9 in DEA-I is executed. We assume that the if block at line 9 in DEA-I is one iteration of the first while loop. Thus, on an input $(a,b,c)$, the number of iterations of the first while loop of DEA-I is equal to the number of recursions of the recursive function in the DEA-R algorithm. 

\par Therefore, the analysis on the number of recursive calls in DEA-R, which is given in sections \ref{secPNRC} and \ref{secACAD}, is also applicable to the number of iterations of the first while loop in DEA-I. 

\section{Periodicity in number of recursive calls in DEA-R}\label{secPNRC}

In this section, we analyse the number of recursive calls by function $f$ of the DEA algorithm. Note that, by number of recursive calls in DEA, we are referring to the number of iterations in DEA-I as well as the number of recursive calls in DEA-R.

\par First, we see the relation between the values of $c$ and the number of recursive calls by the function  $f$ of the DEA algorithm. Then, we define the number of recursive calls as a function of $c$ on fixed values of $a$ and $b$. Let $$c_i=a_i+Q_ia_{i+1}$$ here $1 \le i \le k$ and $a_{i+1}<a_i$, $Q_i$ is any integer (independent of index $i$, $c_i$, $a_i$ or $a_{i+1}$). Description of the notations, $a_i$ and $a_{i+1}$, is given in the notations section. In other words, we are saying that $c_i \in \mathbf{c_i}$ where $\mathbf{c_i}$ is the set of integers expressible as $a_i+Q_ia_{i+1}$ (described in the notations section also). For a $c$, which belongs to all the sets in $\{\mathbf{c_{j_1}},.....,\mathbf{c_{j_l}}\}$, where $j_1<.......<j_l$, DEA-R algorithm makes $j_1$ or $min\{j_1,j_2,......,j_l   \}$ recursive call/calls.
 The algorithm incurs $j_1$ recursive calls because the condition in line no. 6 of the  DEA-R algorithm evaluates to true in the $j_{1}^{th}$ recursive call. To understand this, consider the value, $c-a=c-a_{j_1}$, which is computed at line no. 6 in $j_1^{th}$ recursive call of the DEA-R algorithm. Since $c$ is in the set, $\mathbf{c_{j_{1}}}$, $$c-a_{j_1}=a_{j_1}+Q_{j_1}a_{j_1+1}-a_{j_1}=Q_{j_1}a_{j_1+1}$$
 Then, the condition at line no. 6 of DEA-R specifically checks whether $Q_{j_1}a_{j_1+1} \bmod a_{j_1+1}$ is equal to $0$ or not. This is equal to $0$ and consequently, the condition evaluates to true. Along with this, note that the $j_1^{th}$ recursive call is hit by the DEA-R algorithm, before any of the $i^{th}$ recursive call, where $i$  belongs to the set, $\{j_2,......,j_l\}$. Hence, the DEA-R algorithm executes its first return statement in the $j_1^{th}$ recursive call,  from line no. 8. Thus DEA-R algorithm will incur only $j_1$ recursive calls.  
 This observation on the execution of the DEA-R algorithm is explained in \cite{deora2023algorithm} also. 
We can infer from this observation that if the value of $c$ (third parameter) in the input belongs to all the sets $\mathbf{c_i}$ ($1 \le i\le k$), then the DEA-R algorithm will incur one recursive call. In the next section,  in Theorem \ref{theo11}, we prove that we can find such a value of $c$.  

\par We define a function, $R_{a,b}(c)$, that represents the number of recursive calls for the DEA-R algorithm to solve $ax+by=c$. If $c_i \in \mathbf{c_i}$ and $c_i \notin \mathbf{c_j}$ for any $j<i$,  then $R_{a,b}(c_i)=i$. In Theorem 1, we show that the function $R_{a,b}(c)$ is periodic and find its period. 
\par Before the theorem, we write some assumptions on the input variables of the DEA-R algorithm.  We assume that the equation $a_1x+a_2y=c$ is not solvable in integers. Integers $a_1,a_2,.....,a_{k+1}$ are observed during the execution of the DEA-R algorithm. Specifically, when we execute DEA-R algorithm on the inputs $(a_1,a_2,c)$, then inputs to the subsequent recursive calls in DEA-R algorithm are $(a_2,a_3,c),(a_3,a_4,c),\\.......,(a_k,a_{k+1},c),(a_{k+1},0,c)$. Note that, if we run Euclid's gcd algorithm on $a_1,a_2$, then in the first recursive call the input will be $(a_1,a_2)$, in the next recursive call it is $(a_2,a_3)$, then $(a_3,a_4)$ in the next and so on up to $(a_{k+1},0)$. This assumption is also given in the notations section. Thus $a_3,a_4,......, a_{k+1}$ is the sequence of remainders observed during the execution of Euclid's algorithm. 
\newline  
\begin{theorem} \label{theo10}

 Assume that $R_{a,b}(c)$ is a function from all the possible values of $c$ to the number of recursive function calls incurred by the DEA-R algorithm. The fundamental period of $R_{a,b}(c)$ will be $lcm(a_2,a_3,......,a_{k+1})$.

\end{theorem}

\begin{proof}
Assume that the Diophantine equation $ax+by=c$ has integer solutions. Consider the computation at line no. 6 of DEA-R algorithm, given by 
$$
(c-a_i) \bmod a_{i+1}
$$
It must evaluate to 0 for at least one value of $i$. 
This is equivalent to the statement that, if the Diophantine equation $ax+by=c$ is solvable in integers, then $c$ must belong to at least one of the sets, $\mathbf{c_1},\mathbf{c_2},....,\mathbf{c_k}$. All the integers in a set $\mathbf{c_i}$ are expressible as $a_i+Q_ia_{i+1}$ for any integer $Q_i$. In any integer $c\in \mathbf{c_i}$, we have to add at least $a_{i+1}$, so that $c+a_{i+1}$ is in $\mathbf{c_i}$ . If we consider all the values of $c$ for which Diophantine equation $ax+by=c$ is solvable in integers, we have to add atleast $lcm(a_2,a_3,.......,a_{k+1})$ in $c$ such that if $c\in c_i$, then $c+lcm(a_2,a_3,.......,a_{k+1})\in \mathbf{c_i}$. So, if $c \in \mathbf{c_i}$ but $c \notin \mathbf{c_j}$ for any $j<i$, then $c+lcm(a_2,a_3,.......,a_{k+1})\in \mathbf{c_i}$ and $c+lcm(a_2,a_3,.......,a_{k+1})\notin \mathbf{c_j}$ for any $j<i$. It implies that if DEA-R algorithm incurs $i$ recursive calls on input $a,b,c$, i.e. $R_{a,b}(c)=i$ then $R_{a,b}(c+lcm(a_2,a_3,.......,a_{k+1}))=i$.  Thus, $lcm(a_2,a_3,.......,a_{k+1})$ is a period of the function $R_{a,b}$. 
\par Now, we consider the unsolvable instances of $a,b,c$. If a $c$ does not belong to any sets $\mathbf{c_1}, \mathbf{c_2},.....,\mathbf{c_k}$, then $ax+by=c$ has no integer solutions. On the non-solvable instances of $a,b,c$, the DEA-R algorithm incurs $k+1$ recursive calls, i.e. $R_{a,b}(c)=k+1$. We can say that, $$R_{a,b}(c+lcm(a_2,a_3,.......,a_{k+1}))=k+1$$ Suppose if this is not correct, then $c+lcm(a_2,a_3,.......,a_{k+1})$ must belong to one of the sets, $\mathbf{c_1}, \mathbf{c_2},.....,\mathbf{c_k}$. So, we will be able to write $c+lcm(a_2,a_3,.......,a_{k+1})=a_i+Q_ia_{i+1}$ for any integer $Q_i$. Thus, we will be able to write $c$ also as the similar linear combination of integers $a_i$ and $a_{i+1}$. Consequently, $c$ must belong to the same sets, to which  $c+lcm(a_2,a_3,.......,a_{k+1})$ belong. So, this is a contradiction and for unsolvable instances of $a,b,c$ also, $R_{a,b}(c)=R_{a,b}(c+lcm(a_2,a_3,.......,a_{k+1}))$. 
\par Hence, we proved that $lcm(a_2,a_3,.......,a_{k+1})$ is a period of the function $R_{a,b}$. But, to prove that this is the fundamental period, we have to consider the case when there is a set $\mathbf{c_i}$, such that for all $c \in \mathbf{c_i}$, there is a set $\mathbf{c_j}$ such that $c \in \mathbf{c_j}$, where $j<i$. In this case, the period of $R_{a,b}$ may be smaller than $lcm(a_2,a_3,.......,a_{k+1})$, because $R_{a,b}$ will never be equal to $i$, so that we can divide the $lcm(a_2,a_3,.......,a_{k+1})$ by $a_{i+1}$ to get a smaller period. We will prove that this scenario is not possible for $i<k$ i.e. there exists atleast one value in each of the sets $\mathbf{c_1},\mathbf{c_2},..., \mathbf{c_{k-1}}$ which does not belong to any smaller indexed set. Specifically, we will prove that for all $i<k$, there exists a $c$ with $c \in \mathbf{c_i}$ but $c \notin \mathbf{c_j}$ for all $j<i$. After this we will see that for set $\mathbf{c_k}$, we do not need an integer which belongs to only the set $\mathbf{c_k}$.    
    \par  Now we will prove that, for all $i<k$, $c=a_{i+2}$ belongs to $\mathbf{c_i}$ but it does not belong to $\mathbf{c_j}$ for all $j<i$. Consider the integer $$c=a_{i+2}$$
$$\implies c = a_i-q_ia_{i+1}$$ 

where $q_i$ (defined in notations section) is the quotient of division $\frac{a_i}{a_{i+1}}$.
$$\implies c \in \mathbf{c_i}$$
because all the integers belonging to $\mathbf{c_i}$ are in the form $a_i+Q_ia_{i+1}$ (see notations Section).
Suppose that $a_{i+2} \in \mathbf{c_j}$ for some $j<i$, then 
\begin{equation}  \label{perEq1}
a_{i+2}=a_j+Q_{j}a_{j+1}    
\end{equation}

for an integer $Q_j$, which is positive, negative or $0$. 
\par $j<i$ implies that  $a_i<a_j$ and since $a_{i+2}<a_i$, so $a_{i+2}<a_j$ and $a_{i+2}<a_{j+1}$. Since $a_j$ and $a_{j+1}$ are positive integers, for Equation (\ref{perEq1}) to hold, $Q_j$ must be negative. We write equation \ref{perEq1} as follows:
\begin{equation}\label{perEq2}
    a_{i+2}-Q_ja_{j+1}=a_j
\end{equation}
Since $-Q_j$ is a positive integer and $a_{i+2}<a_{j+1}$, $a_{i+2}$ is the remainder of the division, $\frac{a_j}{a_{j+1}}$. It is a contradiction because we know that $a_{i+2}$ is the remainder of the division, $\frac{a_i}{a_{i+1}}$. Thus for all $i<k$, $a_{i+2}\in \mathbf{c_i}$ but $a_{i+2} \notin \mathbf{c_j}$ for all $j<i$. 
\par Suppose that we have a value in $\mathbf{c_k}$ which belongs to only $\mathbf{c_k}$ and no other sets $\mathbf{c_j}$ such that $j<k$. Then the fundamental period of $R_{a,b}$ will be $lcm(a_2,a_3,.....,a_{k+1})$. Now, assume that all the values in set $\mathbf{c_k}$ belong to any of the sets $c_j$ for some $j<k$. Then fundamental period of $R_{a,b}$ will be $lcm(a_2,a_3,......,a_{k})$. Since $a_{k+1}$ is the $gcd$ of $a_2$ and $a_3$, $lcm(a_2,a_3,......,a_{k+1})=lcm(a_2,a_3,.......,a_k)$. Thus the fundamental period of function $R_{a,b}$ is equal to $lcm(a_2,a_3,......,a_{k+1})$.
\end{proof}

The analysis on the number of recursive calls in the DEA-R algorithm is also applicable to the number of iterations of the first while loop in the DEA-I algorithm.

 \section{Average case analysis of DEA-R}\label{secACAD}
 
In the average case, number of steps in Euclidean GCD algorithm and binary gcd algorithm is $O(\log n)$ \cite{knuth2014art}. They can be extended to solve two-variable linear Diophantine equations. Now, we do the average case analysis of the DEA-R algorithm. 
 \par The DEA-R algorithm has three input parameters, given as $a,b$, and $c$. 
The number of recursive calls for executing the DEA-R algorithm at an input triplet ($a,b,c$) depends on all these parameters. We select the third parameter $c$ of the algorithm to perform the average case analysis of the number of recursive function calls incurred by the algorithm.  So, we consider all the possible integer values of $c$ between $-\infty$ and  $+\infty$ for finding the average number of recursive function calls. We perform the analysis only on $c$, because we can find the average case analysis of Euclid's algorithm (based on $a$ and $b$) in \cite{dixon1970number}, \cite{heilbronn1969average} and \cite{hensley1994number}. We can combine the analysis of Euclid's algorithm with ours to arrive at an analysis that considers all the parameters. Specifically, we can use the average value found by \cite{dixon1970number}, \cite{heilbronn1969average} and \cite{hensley1994number} instead of $\log b$, wherever we use $\log b$ in our analysis. First, we will explain our analysis using an example, and then we will generalize it.

\par The Execution of the DEA-R algorithm for different values of $c$ but fixed values of $a$ and $b$ exhibits the advantage of the DEA-R algorithm over other existing algorithms. As shown in the previous section, the DEA-R algorithm may generate a solution of a  Diophantine equation using only one recursive function call, but we can not know a priori the $c$ value for which it will take 1 recursive call. To foretell any information about the number of recursive calls by the DEA-R algorithm on an input  $(a,b,c)$, it is required to run Euclid's algorithm on the input $(a,b)$. After this, we get all the intermediate values of $(a,b)$, then we can find the specific probabilities, with which the DEA-R algorithm will solve the equation $ax+by=c$  in $i$ recursive calls. For example, consider the equation 
\begin{equation}
\label{egeq1}    1759x+550y=c
\end{equation}  Running  Euclid's algorithm on $(1759,550)$ will follow the steps $(1759,550) \Rightarrow (550,109) \Rightarrow (109,5)\Rightarrow (5,4)\Rightarrow(4,1)$. If $c=1759+550Q_1$, then the DEA-R algorithm will solve the equation \ref{egeq1} in one function call because $c \in \mathbf{c_1}$.  If we assume that $c$ takes values between 1 to a positive integer $L$, the probability that the algorithm will take one recursive call will be
$$P(R_{1759,550}(c)=1)=\frac{\frac{L}{550}+1}{L}$$
This is because $1759=550\times 3+109$, which implies
$$1 \le 1759+ 550Q_1 \le L \implies \left\lceil \frac{1-109}{550} \right\rceil \le Q_1 \le \left \lfloor \frac{L-109}{550}\right\rfloor$$
Since $109 < 550$, so $\lfloor\frac{L-109}{550}\rfloor=\lfloor\frac{L}{550}\rfloor$. Therefore, $0 \le Q_1 \le \lfloor\frac{L}{550} \rfloor$ or $Q_1$ can take $\lfloor\frac{L}{550} \rfloor+1$ i.e. $\frac{L}{550}+1$ \newline($L=lcm(550,109,5,4)$) values between 1 and $L$.

Function $R_{a,b}$ is defined in the previous section. $R_{a,b}(c)$  represents the number of recursive calls by the function $f$ of DEA-R on input $(a,b,c)$. 
 By theorem \ref{theo10}, fundamental period of $R_{1759,550}(c)$ is $lcm(550,109,5,4)$. Since after the fundamental period, the value of the function $R_{1759,550}$ repeats, we assume that $L=lcm(550,109,5,4)$. We have  proved formally that we can take the average over only fundamental period of $R_{a,b}$ to find the average over the entire number line (Appendix \ref{EVNRCDA}).
The DEA-R algorithm will take 2 recursive calls for all $c \in \mathbf{c_2}\setminus \mathbf{c_1}$. As explained in section 2, $\mathbf{c_1}$ and $\mathbf{c_2}$ are the sets of values of $c$, which can be written as $1759+550Q_1$ and $550+109Q_2$ ($Q_1$ and $Q_2$ are integers) respectively. If $c$ belongs to both the sets $\mathbf{c_1}$ and $\mathbf{c_2}$, then the DEA-R algorithm will take only one recursive call, not two. So, we find the cardinality of set $\mathbf{c_1}\cap \mathbf{c_2}$ or solve for the number of integral solutions of $1759+550Q_1=550+109Q_2$. After solving this equation using extended Euclid's algorithm or DEA-R algorithm we can see that the general solution of this equation will be $(220+109m,1099+550m)$ or $(111+109m,549+550m)$. Then, we subtract the cardinality of the set, $\mathbf{c_1} \cap \mathbf{c_2}$, from the cardinality of the set $\mathbf{c_2}$. Hence
$$P(R_{1759,550}(c)=2)=\frac{1}{L}\left(\frac{L}{109}+1-\left(\frac{L}{550\times 109}+1\right)\right)$$ 
where $L=lcm(550,109,5,4)$.
DEA-R algorithm will take 3 recursive calls for $c \in (\mathbf{c_3}\setminus \mathbf{c_2}) \setminus \mathbf{c_1}$. To proceed, we need to compute $|\mathbf{c_3} \cap \mathbf{c_2}| + |\mathbf{c_3}\cap \mathbf{c_1}|-|\mathbf{c_3}\cap \mathbf{c_2}\cap\mathbf{c_1}|$ and then subtract this value from the total number of values of $c$ in the set, $\mathbf{c_3}$. Thus, we find the probability
$P(R_{1759,550}(c)=3)$. Similarly, we find the probabilities, $P(R_{1759,550}(c)=4)$ and $P(R_{1759,550}(c)=5)$. So we can find average number of recursive calls as follows:\begin{multline}R_{1759,550}^{avg}=1P(R_{1759,550}(c)=1)+2P(R_{1759,550}(c)=2)+\\3P(R_{1759,550}(c)=3)+4P(R_{1759,550}(c)=4)+\\5P(R_{1759,550}(c)=5)\end{multline}
Now, we generalize this analysis for any value of $a$ and $b$. In Section \ref{SecEXC}, we propose Theorem \ref{theo11}, which is on the existence of the intersection of sets ($\mathbf{c_1},.....,\mathbf{c_k}$). This theorem can be useful in generalizing the analysis of the average value of $R_{a,b}$, because we need the cardinality of multiple set intersections (for an intersection of any possible combination from $\mathbf{c_1},....,\mathbf{c_k}$).

\subsection{Average number of recursive calls}
 Between 1 to an integer $L$, let $n_1$ denote the number of values of $c$, which are in the form of $\mathbf{c_1}$. By theorem 1, the fundamental period of $R_{a,b}$ is $LCM(a_2,a_3,\ldots,a_{k+1})$. After the fundamental period of the function $R_{a,b}(c)$, its values repeat; therefore it is feasible to assume that $L=LCM(a_2,a_3,\ldots,a_{k+1})$. Since $R_{a,b}$ repeats after the fundamental period, $L=LCM(a_2,a_3,\ldots,a_{k+1})$, average of $R_{a,b}$ over entire number line is equal to its average over the fundamental period (Appendix \ref{EVNRCDA}). For $1<i\le k$, let $n_i$ denotes the number of values of $c$ which belong to the set, $\mathbf{c}_i$ and not to any of the sets, $\mathbf{c_1},\mathbf{c_2},....,$ to $\mathbf{c_{i-1}}$. 
To find $n_1$, we have to solve for $Q_1$ in: 
$$a_1+Q_1a_2=L$$

$n_1=Q_1=\frac{L}{a_2}+1$. 
To find $n_2$, we have to solve for $Q_2$ in
$a_2+Q_2a_3=L$. Then subtract the number of $c$ values in $\mathbf{c_1}$, which are counted in $Q_2$ also i.e. we solve for $c \in \mathbf{c_2}\setminus \mathbf{c_1}$. Similarly, we need to compute $n_3$ and so on up to $n_k$. Assume that the number of $c$ values, for which the Diophantine equation is not solvable, is $n'$. For these values of $c$, the DEA algorithm will incur $k+1$ recursive calls. In the first $k$ recursive calls, the algorithm finds that the value of $c$ does not belong to any of $\mathbf{c_1},...,\mathbf{c_k}$. Then, in the $(k+1)^{th}$ call, it declares that the equation is not solvable.
Then the number of recursive calls by the DEA algorithm in the average case will be as follows:

\begin{equation} \label{avgeq1}
    R_{a,b}^{avg}=\frac{1}{L}(1n_1+2n_2+3n_3+.......+kn_k+(k+1)n')
\end{equation}

Computing an asymptotic value of $R_{a,b}^{avg}$ requires an analysis of 
$a_1,a_2,.....,a_k$ in euclidean algorithm or DEA algorithm. 
We can find an upper bound on $R_{a,b}^{avg}$ by considering only the first term while evaluating $n_i$s. Along with this, if we assume that for all the values of $c$, the Diophantine equation has an integral solution, then $n'=0$.
In that case, we derive the upper bound as follows:

\begin{equation} \label{bound1}
    R_{a,b}^{avg}\le \frac{1}{L}\left(1\left(\frac{L}{a_2}+1\right)+2\left(\frac{L}{a_3}+1\right)+.......+k\left(\frac{L}{a_{k+1}}+1\right)\right)
\end{equation} 

where $L=lcm(a_2,a_3,......,a_{k+1})$. 
\par If $n' \ne 0$, then $n'=L-\frac{L}{a_{k+1}}$. $a_{k+1}$ is the $gcd$ of $a_1$ and $a_2$, so $L-\frac{L}{a_{k+1}}$ is the number of $c$ values between 1 to $L$, for which the Diophantine equation is not solvable.
Then we rewrite the inequality \ref{bound1} as follows:
\begin{multline}\label{bound2}
      R_{a,b}^{avg}\le \\\frac{1}{L}\left(1\left(\frac{L}{a_2}+1\right)+2\left(\frac{L}{a_3}+1\right)+.......+k\left(\frac{L}{a_{k+1}}+1\right)+(k+1)(L-\frac{L}{a_{k+1}})\right) 
 \end{multline}

The assumption that $L$ is a large integer leads to an exact value of the bound on $R_{a,b}^{avg}$, dependent only on inputs. We compute a tight upper bound on $R_{a,b}^{avg}$ from inequality \ref{bound2} as follows:

\begin{multline}R_{a,b}^{avg} \le \lim_{L \to \infty} \frac{1}{L}\left(1\left(\frac{L}{a_2}+1\right)+2\left(\frac{L}{a_3}+1\right)+.......+k\left(\frac{L}{a_{k+1}}+1\right)\right)+\\\frac{1}{L}\left((k+1)(L-\frac{L}{a_{k+1}}) \right)\end{multline}
$$=\frac{1}{a_2}+\frac{2}{a_3}+......+\frac{k}{a_{k+1}}+(k+1)\left(\frac{a_{k+1}-1}{a_{k+1}}\right)$$

\begin{equation}\label{bound3}
 \implies R_{a,b}^{avg} \le  \frac{1}{a_2}+\frac{2}{a_3}+......+\frac{k}{a_{k+1}}+(k+1)\left(\frac{a_{k+1}-1}{a_{k+1}}\right) 
\end{equation}
The upper bound (\ref{bound3})is the tightest upper bound on $R_{a,b}^{avg}$.
In lemma \ref{lem2}, we derive a lower bound on the values of $a_2,a_3,.....,a_{k+1}$. Using this lower bound on the values of $a_2,a_3,.....,a_{k+1}$, we derive a looser upper bound than bound (\ref{bound3}). Lemma \ref{lem2} is stated in \cite{cormen2009introduction} and its proof is given in \cite{Cormen_Sol_ref}. We prove Lemma \ref{lem2} differently than  \cite{Cormen_Sol_ref}.
\par With the aid of \cite[Lemma 31.10]{cormen2009introduction} for a lower bound on $a_1,a_2$, we find a new lower bound on these values. According to \cite[Lemma 31.10]{cormen2009introduction} 
if Euclid's algorithm incurs $k$ recursive calls on inputs $(a_1,a_2)$, then the value of $a_2$ will be at least $F_{k+1}$ and value of $a_1$ will be at least $F_{k+2}$ \cite[Lemma 31.10]{cormen2009introduction}. $F_{k+1}$ and $F_{k+2}$ are the  $(k+1)^{th}$ and $(k+2)^{th}$  fibonacci numbers respectively. Fibonacci numbers are defined by the recursion, $F_n=F_{n-1}+F_{n-2}$ and $F_{1}=0,F_2=1$. 
\newline
\begin{lem}\label{lem2}
If $a_1>a_2\ge 1$ and Euclid's algorithm performs $k$ recursive calls on the inputs $(a_1,a_2)$, then $a_1\ge gcd(a,b).F_{k+2}$ and $a_2 \ge gcd(a_1,a_2).F_{k+1}$

\end{lem}
\begin{proof}
   If $(a_1,a_2),(a_2,a_3),......,(a_k,a_{k+1})$ are the sequence of input arguments observed in Euclid's gcd algorithm. Note that for all $i\le k-1$ 
\begin{equation}\label{lem2eq1}
a_i=q_ia_{i+1}+a_{i+2}   
\end{equation}
    Here $q_i$ and $a_{i+2}$ are the quotients and remainders of division of $a_i$ by $a_{i+1}$ respectively.
    We know that the greatest common divisor $(gcd)$ of $a_1$ and $a_2$ is $a_{k+1}$. So we can divide all integers $a_1,a_2,......,a_{k+1}$ by $a_{k+1}$. All equations in the set of equations given by \ref{lem2eq1} are still satisfied if we divide all of them by $a_{k+1}$. So, when we pass the input arguments $(a_1/a_{k+1},a_2/a_{k+1})$ to the Euclidean algorithm, it still takes $k$ recursive calls. Consequently $a_1/a_{k+1} \ge F_{k+2}$ and $a_2/a_{k+1} \ge F_{k+1}$. Hence $a_1\ge gcd(a,b).F_{k+2}$ and $a_2 \ge gcd(a_1,a_2).F_{k+1}$.
\end{proof}

We modify the inequality \ref{bound3} using lower bound on $a_2,a_3,\\.....,a_{k+1}$ from lemma \ref{lem2}. We arrive at a tighter bound on $R_{a,b}^{avg} $ as follows
$$
R_{a,b}^{avg} \le  \frac{1}{F_{k+1}a_{k+1}}+\frac{2}{F_ka_{k+1}}+......+\frac{k}{F_1a_{k+1}}+(k+1)\left(\frac{a_{k+1}-1}{a_{k+1}}\right) 
$$

By using approximation of $F_k$ as $\frac{\phi^k}{\sqrt{5}}$, where $\phi$ is an irrational number known as golden ratio \cite{cormen2009introduction}.
$$
R_{a,b}^{avg} \le  \frac{\sqrt{5}}{\phi^{k+1}a_{k+1}}+\frac{2\sqrt{5}}{\phi^k a_{k+1}}+......+\frac{k\sqrt{5}}{\phi a_{k+1}}+(k+1)\left(\frac{a_{k+1}-1}{a_{k+1}}\right)   
$$

\begin{multline}    
 \label{bound4}
\implies R_{a,b}^{avg} \le \frac{\sqrt{5}}{\phi^{k+1}a_{k+1}}\left(1+2\phi+3\phi^2+........+k\phi^k\right)+\\(k+1)\left(\frac{a_{k+1}-1}{a_{k+1}}\right)  
\end{multline}

The inequality (\ref{bound4}) has an arithmetic-geometric progression. We rewrite the inequality after computing the sum of the arithmetic-geometric progression for $k+1$ terms.

\begin{multline}\label{bound 5}
    R_{a,b}^{avg} \le \frac{\sqrt{5}}{\phi^{k+1}a_{k+1}}\left(\frac{1-(k+1)\phi^{k+1}}{1-\phi}+\phi\frac{1-\phi^k}{(1-\phi)^2}\right)+\\(k+1)\left(\frac{a_{k+1}-1}{a_{k+1}}\right) 
\end{multline}

\begin{multline*}
=\frac{\sqrt{5}}{\phi^{k+1}a_{k+1}}\left(\frac{1-\phi-(k+1)\phi^{k+1}+(k+1)\phi^{k+2}+\phi-\phi^{k+1}}{(1-\phi)^2}\right) \\+(k+1)\left(\frac{a_{k+1}-1}{a_{k+1}}\right) 
\end{multline*}

$$
=\frac{\sqrt{5}}{a_{k+1}}\left(\frac{1}{(1-\phi)^2\phi^{k+1}}+\frac{k(\phi-1)+\phi-2}{(1-\phi)^2}\right)+(k+1)\left(\frac{a_{k+1}-1}{a_{k+1}}\right) 
$$
The summations that don't contain $k$ can be neglected.  So, we provide an upper bound for the above expression as follows.
$$
\le \frac{\sqrt{5}}{a_{k+1}}\frac{k(\phi-1)}{(\phi-1)^2}+ (k+1)\left(\frac{a_{k+1}-1}{a_{k+1}}\right) 
$$
i.e.
\begin{equation}\label{bound6}
    R_{a,b}^{avg}\le  \frac{\sqrt{5}k}{{(\phi-1)}a_{k+1}}+(k+1)\left(\frac{a_{k+1}-1}{a_{k+1}}\right) 
\end{equation}

We use approximate values of the golden ratio ($\phi$) and $\sqrt{5}$ to write the above inequality as

$$
    R_{a,b}^{avg}\le  \frac{1.414k}{{0.618}a_{k+1}}+(k+1)\left(\frac{a_{k+1}-1}{a_{k+1}}\right) 
$$
i.e.
\begin{equation}\label{bound7}
     R_{a,b}^{avg}\le  \frac{2.28k}{a_{k+1}}+(k+1)\left(\frac{a_{k+1}-1}{a_{k+1}}\right) 
\end{equation}
 When we solve inequality (\ref{bound7}), we can see that the average number of recursive calls taken by the DEA algorithm is $O(k+1)$, i.e., $O(\log a_2)$, or $O(\log b)$. The asymptotic number of recursive calls taken by the Extended Euclid's algorithm is $\Theta(\log b)$. 
 \par In some instances of the problem, we may already know that the equation is solvable, for example, when $a$ and $b$ are coprime. When it is known that the equation is solvable, we get a constant improvement in the average number of recursions. To derive the constant, improvement we consider only $\frac{L}{a_{k+1}}$ (or $\frac{L}{gcd(a,b)}$) solvable instances of input $(a,b,c)$. So, we find the average number of recursive calls over only these $\frac{L}{a_{k+1}}$ inputs. By repeating the derivation starting from inequality (\ref{bound1}), we find that the $R_{a,b}^{avg}$ is $O(2.28k)$ or $O(2.28 (\log b - 1))$. This is a constant improvement over $\Theta(\log b)$ recursive calls by Extended Euclid's algorithm. 
 
 \subsection{An upper bound based on Fibonacci numbers} We get $2.28(\log b -1)$ as an upper bound on the value of $R_{a,b}^{avg}$. If we compute exact value rather than an upper bound, it will contain a multiplier less than 1 in place of 2.28. It will be less than 1 because for all values of $c$ the DEA-R algorithm will incur number of recursive calls which is between $1$ and $k$ (both inclusive). So, to find a more accurate bound we will consider the intersection of sets $\mathbf{c_i}$
also while computing the Equation \ref{avgeq1}. We rewrite the Equation (\ref{avgeq1}) again for the solvable instances only.
\begin{equation}
\label{avgAgain}   R_{a,b}^{avg}=\frac{1}{L}(1n_1+2n_2+3n_3+\ldots+kn_k)
\end{equation} As noted earlier, to compute $n_i$ for any $i>1$, we have to use the Inclusion-Exclusion principle. $n_i$ is the number of values of $c$ between $1$ to $L$ on which DEA-R algorithm incurs $i$ recursive calls. To compute $n_i$ we find those values of $c$ which are in $\mathbf{c_i}$ and not in any of $\mathbf{c_j}$ such that $j<i$. Any value of $c$ which is in both $\mathbf{c_i}$ and $\mathbf{c_j}$ will be in the form $a_i+Q_ia_{i+1}$ and $a_j+Q_ja_{j+1}$. If we add any multiple of $lcm(a_{i+1},a_{j+1})$ in that value of $c$, we get another value in the intersection of $\mathbf{c_i}$ and $\mathbf{c_j}$. We assume that $\mathbf{c_i}\cap \mathbf{c_j}$ is not empty. So, to find the number of values of $c$ which are in the intersection  $\mathbf{c_i}\cap \mathbf{c_j}$, we divide $L$ by $lcm(a_{i+1},a_{j+1})$. Similarly for  finding number of values in intersection, $\mathbf{c_i}\cap\mathbf{c_j}\cap \mathbf{c_l}$ (if it is not empty) we divide $L$ by $lcm(a_{i+1},a_{j+1},a_{k+1})$. Using the Inclusion-Exclusion principle, we will compute $n_i$ as follows:
$$n_i=\frac{L}{a_{i+1}}-\sum_{1<j<i+1}\frac{L}{lcm(a_{i+1},a_j)}+\sum_{1<j,k<i+1}\frac{L}{lcm(a_{i+1},a_j,a_k)}+\ldots$$
Here we have assumed all the intersections of the sets $\mathbf{c_i}$'s is non-empty.
 We use Lemma \ref{lem2} for a lower bound on the $a_i$'s, i.e.
\begin{multline}\label{bound 14}
n_i\le\frac{L}{F_{k+3-{(i+1)}}}-\sum_{1<j<i+1}\frac{L}{lcm(F_{k+3-{(i+1)}},F_{k+3-j})}+\\\sum_{1<j,l<i+1}\frac{L}{lcm(F_{k+3-{(i+1)}},F_{k+3-{j}},F_{k+3-l})}+\ldots    
\end{multline}

where $F_i$'s represent $i^{th}$ Fibonacci numbers. We use the inequality given above to rewrite the equation (\ref{avgAgain}) as the follows. When we rewrite it, we cancel out the variable $L$ which is in denominator and numerator both.
\begin{multline}\label{bound 15}
R_{a,b}^{avg}\le\frac{1}{F_{k+1}}+2\left(\frac{1}{F_k}-\frac{1}{lcm(F_k,F_{k+1})}\right)+\\3\left(\frac{1}{F_{k-1}}-\frac{1}{lcm(F_{k-1},F_k)}-\frac{1}{lcm(F_{k-1},F_{k+1})}+\frac{1}{lcm(F_{k-1},F_k,F_{k+1})}\right)
+\ldots \\+k\left(\frac{1}{F_2}-\sum_{1<j<k+1}\frac{1}{lcm(F_2,F_{k+3-j})}+
\sum_{1<j,l<k+1}\frac{L}{lcm(F_{k+3-{(i+1)}},F_{k+3-{j}},F_{k+3-l})}+\ldots\right)\end{multline}
The exact value of right hand side of the inequality (\ref{bound 15}) is a tighter bound than the bound (\ref{bound7}) on the average number of recursive calls taken by DEA-R algorithm. Note that, EEA-R algorithm takes $k+1$ recursive calls. We designed a program to compute this value for different values of $k$. We have provided the details of this computation in the Implementation and result section. 

 \subsection{Existance of an integer $c$ in the intersection of all $\mathbf{c_i}$'s}\label{SecEXC}  In the theorem \ref{theo11}, we propose a sufficient condition for the set $\mathbf{c_1} \cap \mathbf{c_2} \cap ......\cap \mathbf{c_k} \ne \phi $ and find its cardinality. Similarly to what is done in Theorem \ref{theo11}, we can construct a value of $c$ at all possible intersections among the sets $\mathbf{c_1},\mathbf{c_2},.....,\mathbf{c_k}$. We can use the cardinality of these sets to find a tighter upper bound on the inequality \ref{bound1} for the average case analysis of the DEA-R algorithm. Following the theorem \ref{theo11}, we give an example where the condition given in theorem \ref{theo11} is not satisfied but $\mathbf{c_1} \cap \mathbf{c_2} \cap ......\cap \mathbf{c_k} \ne \phi $. Thus, we find that although the condition given in theorem \ref{theo11} is sufficient, but it is not a necessary condition. We present another example of inputs, where $\mathbf{c_1} \cap \mathbf{c_2} \cap ......\cap \mathbf{c_k} = \phi $. \newline\begin{theorem} \label{theo11}
   If $gcd$'s of all two-integer pairs from  $a_1,a_2,....$ $..,a_k,a_{k+1}$ are equal,  then there exists a value of $c$, which is present in all the sets $\mathbf{c_i}$, where $1\le i \le k$. Here each $c$ value in $\mathbf{c_i}$ can be represented as $c=a_i+Q_ia_{i+1}$, for any integer $Q_i$.  $a_1,a_2,.......,a_k$ are integers, such that the inputs in the subsequent recursive function calls in Euclid's gcd algorithm  are in the following order:
$(a_1,a_2)\to (a_2,a_3) \to ...........(a_k,a_{k+1})$.
\end{theorem}

\begin{proof}
Let $a_{k+2}=0$, which is the remainder when we divide $a_k$ by $a_{k+1}$. We use the Chinese remainder theorem to prove the following:
\par
For all
$1\le j < l \le k$

\begin{equation}\label{sysEq1}
\begin{split}
 c & \equiv a_{j+2} \bmod a_{j+1}
\\
c &\equiv a_{l+2} \bmod a_{l+1}
\end{split}
\end{equation}
The system of congruences given by (\ref{sysEq1}) is solvable implies that 
we have a $c$ which belongs to all the sets $\mathbf{c_i}, \,(1\le i \le k)$. If we have an integer $c$, such that $c \equiv a_{i+2} \bmod a_{i+1}$, then $c=q_i'a_{i+1}+a_{i+2}$
$$\implies c= q_i'a_{i+1}+a_i-q_ia_{i+1}$$
where $\left\lfloor {\frac{a_i}{a_{i+1}}}\right\rfloor =q_i$. 
By the Chinese remainder theorem, the system of congruences (\ref{sysEq1}) has a solution if and only if the following congruences can be solved:

\begin{equation} \label{sysEq1_IffCond}
a_{j+2} \equiv a_{l+2} \bmod gcd(a_{j+1},a_{l+1})
\end{equation}
for all
$1\le j < l \le k$.
\par Since every two-integer pairs from $a_1,a_2,.....,a_{k+1}$ have the same $gcd$, the system of congruences given by (\ref{sysEq1_IffCond}) is solvable. Therefore, we have a $c$ which is present in all the sets $\mathbf{c_i}$, where $1 \le i \le k$.

\end{proof}
To compute the cardinality of the set, $\mathbf{c_1} \cap \mathbf{c_2} \cap ......\cap \mathbf{c_k}$, it is required  to find the number of solutions of the system of equations given by (\ref{sysEq1}).    Similarly, we can find the cardinality of all the possible intersections of the sets in $\mathbf{c_1},\mathbf{c_2},.....,\mathbf{c_k}$. Using these values, we can find the exact values of $n_1,n_2,.......,n_k$, respectively, and use them to compute $R_{a,b}^{avg}$ in equation (\ref{avgeq1}).
\par 
In theorem \ref{theo11}, we assume that all the two-integer pairs in $a_1,a_2,.....,a_{k+1}$ have equal gcd, however, this may not always be the case. For example, if the initial inputs are $(89,55)$, then \begin{multline}
    a_1=89,a_2=55, a_3=34,a_4=21, a_5=13,\\a_6=8,a_7=5,a_8=3,a_9=2,a_{10}=1
\end{multline}

Here, the $gcd$'s of all two-integer pairs are not the same. Along with this, we can not find a value of $c$, which belongs to all the sets $c_1,c_2,.....,c_9$. To see this, consider the following Diophantine equation:

$$
89+Q_155=8+Q_65
\implies 55Q_1-5Q_6=-81
$$

It has no integral solution, because 81 is not a multiple of $gcd(55,5)$.
\par Although theorem \ref{theo11} requires an assumption on the input parameters, for the existence of a $c$ value which belongs to all the sets $\mathbf{c_i}, ( 1 \le i \le k)$, but there may be inputs which don't satisfy this assumption and still we can find the required value of $c$. For example, if we consider the input $a=1759,b=550$, then  
\begin{multline}
    a_1=1759, a_2=550, a_3=109, a_4=5, a_5=4, a_6=1
\end{multline}

For these input parameters we can find a value of $c$ which belongs to all the sets $\mathbf{c_1},\mathbf{c_2},\mathbf{c_3},\mathbf{c_4}$ and $\mathbf{c_5}$. So the assumption in theorem \ref{theo11} is a sufficient condition for the existence of $c'$, belonging to $c_i ,(1\le i \le k)$, but it is not a necessary and sufficient condition. 
\par The analysis in this section on the number of recursive calls in the DEA-R algorithm is also applicable to the number of iterations of the first while loop of the DEA-I algorithm.

 \section{Implementation and result}\label{secIR}
 We implemented the  DEA-I algorithm on a computer with a 12th-gen Intel Core i7 1.70 GHz processor.  
  We developed a C program using the GMP library \cite{GMP1} (to handle computation on large integers) for the DEA-I and EEA-I algorithms. 
We worked on three implementation tasks. In the first one, we compared the DEA-I algorithm with existing algorithms. The second implementation task depicts the superiority of the DEA-I algorithm over the Extended Euclid's algorithm. In the second implementation task, we verified that the average number of iterations by the DEA-I algorithm is always less than that by the Extended Euclid's algorithm on 100\% of the solvable inputs. For the unsolvable instances of inputs $(a,b,c)$, DEA-I algorithm required equal number of iterations to that in EEA-I algorithm. In the third implementation task we plotted the bound (\ref{bound 15}) and compare it with average number of recursive calls in EEA-R algorithm.
  \subsection{Average number of iterations in DEA-I and other algorithms}
    In the first implementation task, we compared the DEA-I algorithm with the EEA-I and the extended Euclid's algorithm given on \cite{EEAGit1}, which we refer to as EEA-2. The C language code of EEA-2 given in \cite{EEAGit1} is an implementation of a version of Extended Euclid's algorithm which is given in \cite{knuth2014art}. We compared the number of iterations taken by the first while loop in the DEA-I algorithm with the number of iterations taken by the first while loops of the other two algorithms (EEA-I and EEA-2) on the same inputs ($a,b,c$). Note that the EEA-2 algorithm does not use a stack anywhere in the program, whereas EEA-I uses a stack.  EEA-2 and EEA-I algorithms are used for solving  the Diophantine equation $ax+by=gcd(a,b)$. We extend them to solve $ax+by=c$ for any integer $c$.

We executed the C programs for DEA-I, EEA-2, and EEA-I on 100000 input triplets ($a,b,c$). For each input triplet we chose uniformly random integers ($a,b,c$) where $a$ is between $2$ and $2^{4096}$ (both inclusive), $b$ is between $1$ and $a-1$ (both inclusive) and $c$  is between $1$ and $2^{4096}$(both inclusive). For these inputs, we compared the three algorithms and found that the average number of iterations by the first while loops of the DEA-I, EEA-2, and EEA-I algorithms are 2390.2, 2392, and 2393, respectively. These figures do not show any significant improvement of our algorithm over existing algorithms, but the results of the next implementation task show it.

 \subsection{Plotting number of iterations in DEA-I and other algorithms} 
For our second implementation task, we selected two random integers $a$ and $b$ such that $1 \le b < a \le 2^{512}$. We fixed these integers, ($a,b$) and chose 100000 values of $c$ randomly such that $1\le c \le 2^{512}$. For these values of  ($a,b,c$) values, we executed DEA-I, EEA-I and EEA-2 algorithms and for first 5000 of the ($a,b,c$) values we plotted the number of iterations in each of those algorithms in figure \ref{fig3}. For the remaining inputs we got a similar kind of graph. The number of iterations in DEA-I algorithm is equal to the number of recursions in DEA-R algorithm. This equivalence between recursion and iteration is explained in section \ref{secDEAI}. To show the equivalence between number of recursions in EEA-R and number of iterations in EEA-I algorithm, in the implementation we plotted the number which is one more than the number of iterations by the first while loop of EEA-I algorithm. We do not require a recursive counterpart of the iterative EEA-2 algorithm. So, we did not modify the number of iterations by the first while loop of the EEA-2 algorithm while plotting it. In Figure \ref{fig3}, the x-axis represents serial number of a particular $c$ value. Note that the x-axis does not represent the value of $c$.
\par
The graph (Figure \ref{fig3}) depicts that the number of iterations in DEA-I algorithm is always bounded above by the number of iterations in EEA-I algorithm. The EEA-I and EEA-2 algorithms require constant number of iterations because in both of them, the number of iterations depends upon  the values of $a$ and $b$. The number of iterations by both the algorithms (DEA-I and EEA-I) is same for the inputs for which the equation has no solutions. For solvable equations the DEA-I algorithm always require less number of iterations than that in EEA-I algorithm. We have fixed values of $a$ and $b$ and vary only $c$ to observe the distributions of the function $R_{a,b}$. In section \ref{secPNRC} and section \ref{secACAD}, we saw that the function $R_{a,b}$ depends upon values of $c$ and there are specific probabilities with which $R_{a,b}$ takes a value. The graph \ref{fig3} depicts that the probability of getting a lower value of $R_{a,b}$ is less and probability of getting a higher value of $R_{a,b}$ is high. This distribution of $R_{a,b}$ matches with our theoretical result
on the probability in section \ref{secACAD}. The value of $R_{a,b}$ can be as low as $1$ also (for $c \in \mathbf{c_1}$), but its probability will be the lowest. So in the implementation, we did not observe such a value of $c$ for which $R_{a,b}=1$. We also did not observe many other values,  lower than $302$. This was the average case performance on the implementation of the DEA-I algorithm against the Extended Euclid's algorithm when we vary only $c$. 

\begin{figure}[H]
     \centering
     \includegraphics[width=1\linewidth,height=18cm]{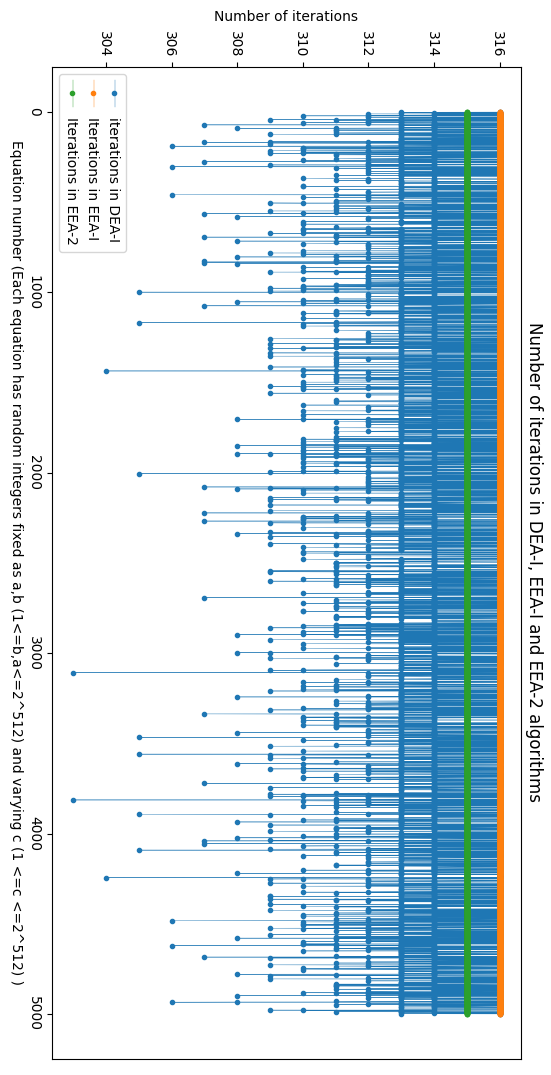}
     \caption{Number of iterations in DEA-I, EEA-I and EEA-2 algorithms on $512-$bit random inputs ($a,b,c$) such that $a$ and $b$ is fixed and $c$ is varying.}
     \label{fig3}
 \end{figure}

\par  We plotted another graph (Figure \ref{fig4}) for random inputs $a,b,c$, where $1 \le b <a\le 2^{512}$ and $1 \le c \le 2^{512}$. Now all the input parameters $a,b,c$ are varying. We took 100000 such random inputs and found the number of iterations in the DEA-I and EEA-I algorithms. For a clear depiction, we plotted the number of iterations for only first 1000 values of $(a,b,c)$. For the remaining values, we got a similar graph. The x-axis represents serial numbers of the Diophantine equations. To increase the visibility in graph, we did not plot number of iterations in EEA-2 algorithm. EEA-2 is an iterative Extended Euclid's algorithm and we do not need to map its iterations to recursions. So,  Number of iterations in EEA-2 algorithm is consistently 1 less than that in the EEA-I algorithm. In this implementation, we found that the maximum difference between the number of iterations in DEA-I and EEA-I algorithms is 10.  

\begin{figure}[H]
     \centering
     \includegraphics[width=1\linewidth,height=18cm]{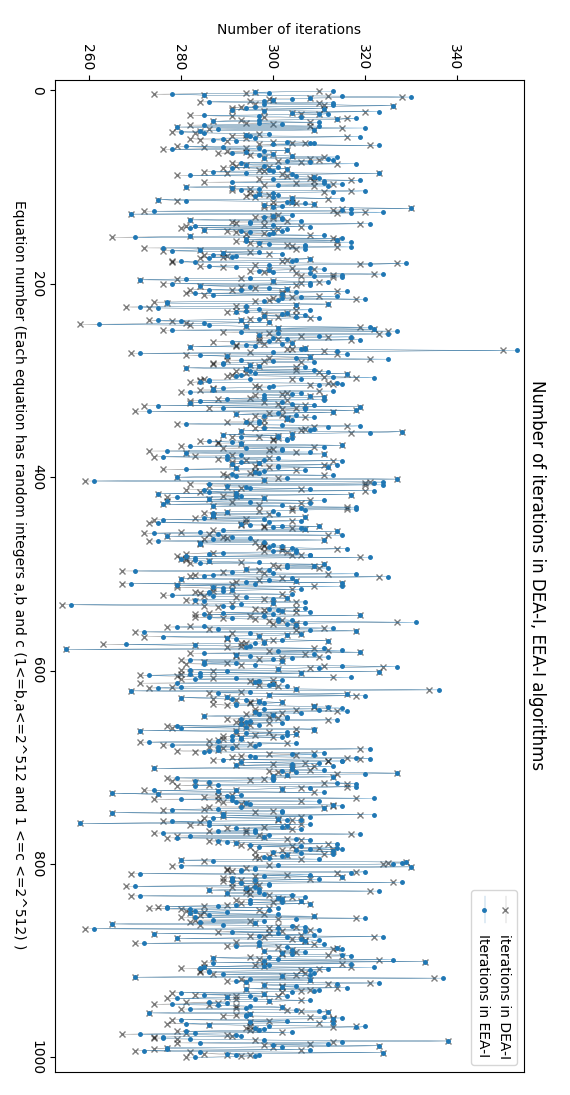}
     \caption{Number of iterations in DEA-I, EEA-I algorithms on $512-$bit random inputs $(a,b,c)$ such that all parameters $a$, $b$ and $c$ is varying.}
     \label{fig4}
 \end{figure}

 \subsection{Theoretical average of number of  recursions in DEA-R and EEA-R algorithm}
 In another implementation task, we designed a C program to compute bound (\ref{bound 15}) against values of $k=1$ to upto $k=33$. With increase in value of $k$, the number of $lcm$ terms increase exponentially. There are exactly $2^k-k-1$, $lcm$ terms in the bound (\ref{bound 15}), which makes it inefficient to compute it for larger $k$. So, we could not compute the bound after $k=33$.  We plotted the bound (\ref{bound 15}) against values of $k+1$ i.e. $\log b$. The plot is shown in the Figure \ref{fig:bound18}. The Bound (\ref{bound 15}) represents a tight theoretical bound on the average number of recursions in DEA-R algorithm on a given  bit - size of $b$. In EEA-R algorithm the average number of recursive calls will be equal to $\log b$. In our implementation we found that the number of recursive calls by DEA-R algorithm is atleast 4.5 less than that by EEA-R algorithm, when the bit - size of $b$ is more than 12.
\begin{figure}[H]
    \centering
    \includegraphics[width=1\linewidth]{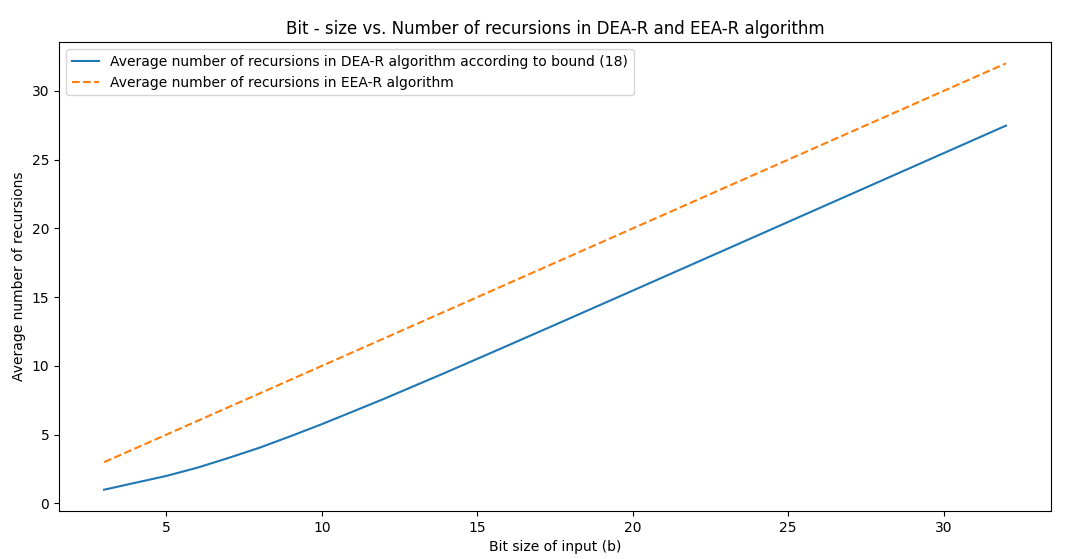}
    \caption{Average number of recursions in DEA-R and EEA-R algorithms with increase in bit - size of $b$}
    \label{fig:bound18}
\end{figure}
 
 \section{Conclusions and Discussions}\label{secCon}
 We found that the number of recursive calls in the DEA-R algorithm or the number of iterations in the first while loop of the DEA-I algorithm on an input $a,b,c$, i.e., $R_{a,b}(c)$ is periodic with respect to $c$. The period is $lcm$ of all the distinct arguments to the recursive function call except $a_1$. The period is used in the average case analysis of DEA-R or DEA-I.  We can try to find the period of $R_{a,b}$ for larger values of $k$ ($k\to \infty$), similar to a result given in  \cite{sanna2022lcm}.   The analysis of periodicity of the function $R_{a,b}$  yields the distribution of the cost of the DEA-R or DEA-I algorithm in the average case. The distribution of cost also depicts the distribution of general forms of $c$ (in the sets, $\mathbf{c_1},\mathbf{c_2},.......,\mathbf{c_k}$). But by this analysis, we are unable to find the exact size of the sets, $\mathbf{c_1},\mathbf{c_2},.......,\mathbf{c_k}$. To know the exact cardinality of these sets, we can use Theorem \ref{theo11} and construct similar sets of theorems. Theorem \ref{theo11} gives us a constructive proof of the existence of a $c$ value, which is present in all the sets, $\mathbf{c_i} \,\,\,\, 1\le i \le k$, under some constraints on the pairwise gcd of intermediate inputs.  In future work, we can focus on another case, when the integers do not have their pairwise gcd equal. 
 The theorem \ref{theo11}  can be further improved by the techniques from operator theory as done for the analysis of Euclid's algorithm in \cite{baladi2005euclidean}. Theorem \ref{theo11} is applicable in finding a closed form expression for average case complexity instead of inequality (\ref{bound2}). 
 
 \par We found multiple bounds on the average number of recursive calls taken by the DEA-R algorithm. We find that the average number of recursive calls for solving the Diophantine equation $ax+by=c$ by the DEA-R algorithm is bounded above by $\theta(\log b)$. For finding a tighter bound, we assume that it is already known to us that the Diophantine equation $ax+by=c$ is solvable. In this case, we fix $a$ and $b$ and find the average of the number of recursive calls over all the values of $c$, which correspond to solvable instances of ($a,b,c$). We find that the DEA-R algorithm takes less than or equal to $2.28 (\log b-1)$ steps. We found another improvement on this bound (bound (\ref{bound 15})), which is related to Fibonacci numbers. All the bounds show that DEA-R algorithm shows a constant improvement over the time complexity of the Extended Euclid's algorithm for solvable instances of inputs.  
 \par The implementation results show that the average number of recursive calls for solving the Diophantine equation $ax+by=c$ by the DEA-R algorithm is less than the existing algorithms. We also plot the bound (\ref{bound 15}) on the average number of recursive calls. It shows that the average value for DEA-R algorithm is less than that for EEA-R algorithm, when bit-size of $b$ is more than 12.  
 \par  It is possible that for the same inputs, the arithmetic operations take more time in the DEA-R algorithm than in existing algorithms, depending on the size of operands inside them. So in spite of the DEA-R or DEA-I algorithm incurring less number of recursive calls or fewer iterations, respectively, the difference between arithmetic operation execution times in the individual recursive calls or iterations can result in the DEA-R or DEA-I algorithm being less efficient than existing algorithms. For example, if the bit size of $c$ is much higher than the bit size of $max(a,b)$, then the time complexity of the DEA-R algorithm or DEA-I algorithm is dominated by $\log c$ instead of $\log(max(a,b))$. In future work, it is possible to reduce the cost of arithmetic operations inside each recursion.     
 \par For computing modulo multiplicative inverse, DEA-R and DEA-I algorithms may not be efficient if we directly call $f(a,b,1)$, but we can design a randomized reduction algorithm for modulo multiplicative inverse. For this, we select two integers, $x'$ and $y'$, randomly and compute $ax'+by'$. If $ax'+by'=c'$, then we use the DEA-R or DEA-I algorithm to solve $ax+by=c'+1$ for $x$ and $y$. The efficiency of the reduction depends upon the value of $c$, which we use for solving $ax+by=c'+1$.  DEA-R and DEA-I algorithms are most efficient when $c$ is in the general form of $\mathbf{c_1}$. 
 
 \par It can be an intriguing area of research to develop an algorithm that randomly selects those values of $x',y'$ in polynomial time, such that $ax'+by'=c_1$ and with high probability $c_1+c$ belongs to set $\mathbf{c_1}$. We apply the DEA-R algorithm to solve the equation $ax+by=c_1+c$. This may lead to developing a probabilistic constant-time algorithm for finding solution of the Diophantine equation $ax+by=c$. Another direction of research may be to find the greatest common divisor ($gcd(a,b)$) along with the solution of  Diophantine equation.  

\par Like any other work, our work also has some limitations. The analysis done in Section \ref{secACAD} does not produce the exact value of $R_{a,b}^{avg}$, but produces upper bounds on it. Theorem \ref{theo11} does not tell us how to find the number of values of $c$ which are present in all the sets $\mathbf{c_i}$, where $1 \le i \le k$, which is needed to find the exact value of $R_{a,b}^{avg}$. 
A constraint (pairwise $gcd$ is equal) in Theorem \ref{theo11} is imposed on $a_1,a_2,.....,a_{k+1}$ for theorem \ref{theo11} to hold. We do not give an estimate of the fraction of random inputs on which this constraint is satisfied. In the near future, we would like to work on these issues also.








\bibliography{sn-bibliography}

\begin{appendices}

\section{Expected value of number of recursive calls by DEA-R algorithm}\label{EVNRCDA}
Let on the entire number line, there be $\alpha$ intervals, each of size $L$. $L$ is the fundamental period of the function $R^{avg}_{a,b}$ (Theorem \ref{theo10}). Let $I$ denotes a random variable representing the $I^{th}$ interval. $I$ can take values between 1 to $\alpha$. Let $J$ is a random variable denoting the number of recursive calls by DEA-R algorithm on a value of $a$ and $b$. We can write $R^{avg}_{a,b}$ as expectation of $J$ as follows:

\begin{equation} \label{appendEq1}
 R^{avg}_{a,b}=\mathbb{E}[J]=\sum_{i=1}^{\alpha}\Pr(I=i)\mathbb{E}[J|I=i]  
\end{equation}

Since size of each interval is equal to fundamental period of $R_{a,b}^{avg}$,

$$
\mathbb{E}[J|I=i]=\mathbb{E}[J|I=i']= \mu \,\,\,(\mbox{let})
$$
for any two distinct intervals $I=i$ and $I=i'$.
Then by Equation (\ref{appendEq1}),
$$
\mathbb{E}[J]=\sum_{i=1}^{\alpha}\frac{1}{\alpha}\mu
$$
$$
\implies \mathbb{E}[J]=\frac{\alpha}{\alpha}\mu
$$
$$
 R^{avg}_{a,b}(c)=\mathbb{E}[J]=\mathbb{E}[J|I=i]=\mu
$$
for any interval $I=i$.

\end{appendices}



\end{document}